\documentclass{article}
\usepackage{spconf,amsmath,graphicx}
\usepackage{algorithm}
\usepackage{algpseudocode}
\usepackage{multirow}

\usepackage{amsthm}

\newtheorem{lemma}{Lemma}[section]

\usepackage{layouts}
\usepackage{tikz}
\usepackage{tkz-graph}
\usepackage{pgfplots}

\newcommand{\mC}{{\rm C}}
\newcommand{\mP}{{\rm P}}

\newcommand{\mI}{{\rm I}}
\newcommand{\mT}{{\rm T}}
\newcommand{\mA}{{\rm A}}
\newcommand{\mb}{{\rm b}}
\newcommand{\ms}{{s}}

\usepackage[square,sort,comma,numbers]{natbib}
\usepackage{amsmath,amssymb,amsfonts}

\title{Privacy-Preserving Distributed Optimisation using Stochastic PDMM}

\name{Sebastian O.\ Jordan$^{1}$, Qiongxiu Li$^{2}$ and Richard Heusdens$^{1,3}$}
\address{$^{1}$ Delft University of Technology, the Netherlands \\$^{2}$ Tsinghua University, China \\$^{3}$ Netherlands Defence Academy, the Netherlands}
\ninept

\begin{document}
\ninept
\addtolength{\abovedisplayskip}{-1.0mm}
\addtolength{\belowdisplayskip}{-1.0mm}
\maketitle
\begin{abstract} 
Privacy-preserving distributed processing has received considerable attention recently. The main purpose of these algorithms is to solve certain signal processing tasks over a network in a decentralised fashion without revealing private/secret data to the outside world. 
Because of the iterative nature of these distributed algorithms, computationally complex approaches such as (homomorphic) encryption are undesired.
Recently, an information theoretic method called subspace perturbation has been introduced for synchronous update schemes. The main idea is to exploit a certain structure in the update equations for noise insertion such that the private data is protected without compromising the algorithm's accuracy. This structure, however, is absent in asynchronous update schemes.
In this paper we will investigate such asynchronous schemes and derive a lower bound on the noise variance after random initialisation of the algorithm. This bound shows that the privacy level of asynchronous schemes is always better than or at least equal to that of synchronous schemes. Computer simulations are conducted to consolidate our theoretical results.
\end{abstract}

\section{Introduction}
\label{sec:introduction}

In recent years, the large increase in connected devices and the data that is collected by these devices has caused a heightened interest in distributed signal processing. As many practical distributed networks are of heterogeneous nature, algorithms operating within these networks need to be simple, robust against network dynamics, asynchronous and resource efficient. Popular applications of distributed signal processing include telecommunication, wireless sensor networks, cloud computing and machine learning \cite{ric:08,zha:13,boy:06}.  
Due to the absence of a central processing point (fusion center), participants/agents/nodes use their own processing ability to locally carry out simple computations and transmit only the required and partially processed data to neighbouring nodes. The transmission of the data could potentially lead to information leakage.
In certain applications like medical systems, financial analysis, and smart grids \cite{giaconi2018privacy}, the data held by the agents/participants is privacy sensitive and needs to be protected from being revealed to the outside world.

One way to protect the data from being revealed is to 
apply some form of encryption, like homomorphic encryption \cite{gentry2009fully}, to encrypt the private data such that the adversary cannot decrypt the encrypted data within a certain period of time. Such approaches e.g. \cite{xu2015secure,zhang2019admm} are often computationally complex, because the encryption functions require high computation overhead.
Alternatively, we could rely on techniques based on information-theoretical principles where the adversary is assumed to have unbounded compute power but does not have sufficient information for inferring the secret/private data. In this paper we will focus 
on information-theoretic approaches because these methods are usually computationally lightweight as they do not require heavy encryption functions. 

Existing information-theoretic approaches for privacy-preserving distributed optimization can be classified into three classes.   First, one can use differential privacy \cite{dwork2006calibrating} to ensure user privacy, see e.g. \cite{huang2015differentially,nozari2018differentially,zhang2018improving,xiong2020privacy}. It guarantees privacy in the worst case that all nodes  in the network except one are not trustworthy. That is, a node’s privacy is ensured to be protected even though all other nodes in the network can collaborate. While this gives strong privacy guarantees, the downside is that it compromises accuracy. Second, for certain distributed optimization problems one can use secret-sharing based techniques \cite{Cramer2015}, like the work proposed in \cite{gupta2017privacy,li2019privacyA,tjell2020privacy}. Secret sharing works by splitting the sensitive information into a number of shares, and the privacy is guaranteed under the assumption that the sensitive information cannot be reconstructed unless a sufficient number of nodes cooperate together, which is assumed to be infeasible. The main advantage is that it achieves privacy at no accuracy cost.  However, it is shown to be only applicable to certain types of problems, and for each problem it requires a tailored adoption. In addition, it often requires the underlying graph to be fully-connected, which limits the scalability of the network in practice. The third privacy-preserving method is the subspace perturbation method \cite{Jane2020ICASSP,Jane2020LS,Jane2020TSP}, which  provides a lightweight yet general solution to ensure privacy by inserting noise in a subspace that is not affected during the iterations. The noise insertion will obfuscate the private data, but does not affect the optimization result. Moreover, the method is generally applicable to a variety of distributed optimization problems.

Although subspace perturbation seems to be a promising method for preserving the privacy, the results in \cite{Jane2020ICASSP,Jane2020LS,Jane2020TSP} were derived for a synchronous update setting. For asynchronous update schemes, however, there does not exist a subspace that is not affected during the iterations. In this paper we will investigate such asynchronous schemes and derive a lower bound on the noise variance that can be used to obfuscate the private data when the probability of a node being updated is uniform. The  bound guarantees that the obtained privacy level of asynchronous schemes will be higher than (at least equal to) the original subspace perturbation approach proposed for synchronous schemes, while the algorithm's accuracy still remains uncompromised.


\section{Notations and problem statement}\label{sec:3dist}
Consider a graphical model $\mathcal{G}=(\mathcal{V},\mathcal{E})$, where $\mathcal{V}$ is the set of vertices representing the nodes/agents in a network and $\mathcal{E}=\{(i,j): i, j\in \mathcal{V}\}$ is the set of (undirected) edges in the graph representing the communication links in the network. We use ${\cal E}_{\rm d}$ to denote the set of all directed edges (ordered pairs). Therefore, $|\mathcal{E}_{\rm d}|=2|\mathcal{E}|$. 
 Hence, given a graph $G=(\mathcal{V},\mathcal{E})$, only neighbouring nodes are allowed to communicate with each other directly. 
We use $\mathcal{N}_i$ to denote the set of all neighbouring nodes of node $i$, i.e., $\mathcal{N}_i=\{j:(i,j)\in \mathcal{E}\}$ and $d_i = |{\cal N}_i|$ denotes the number of neighbouring nodes of node $i$. 

Given a vector $x$, we use $\|x\|$ to denote its $\ell_2$-norm. When $x$ is updated iteratively, we write $x^{(k)}$ to indicate the update of $x$ at the $k$th iteration. When we consider $x^{(k)}$ as a realisation of a random variable, the corresponding random variable will be denoted by $X^{(k)}$ (corresponding capital). 
Furthermore, we use the following notational conventions: a variable $z_i$ is related to node $i$; a variable $z_{i|j}$ is related to edge $(i,j)$ but held by node $i$. 

We would like to solve the following optimisation problem:
\begin{equation}
\begin{array}{ll} \text{minimise} & {\displaystyle \sum_{i\in {\cal V}} f_i(x_i)} \\\rule[4mm]{0mm}{0mm}
\text{subject to} & \mA_{ij}x_i + \mA_{ji}x_j = \mb_{ij}, \quad (i,j)\in \cal  E,
\end{array}
\label{eq:problem}
\end{equation}
where $f_i : \mathbb{R}^{n_i}\mapsto \mathbb{R}\cup\{\infty\}$ are CCP functions, $\mA_{ij}\in\mathbb{R}^{m_{ij}\times n_i}$ and $\mb_{ij}\in\mathbb{R}^{m_{ij}}$.
Let $n=\sum_i n_i$ and $m =\sum_{(i,j)}m_{i,j}$.
For simplicity, we will assume that $n_i=n_j = m_{ij}$, $\mA_{ij}=-\mA_{ji} = \mI$ if $i > j$, and $\mb_{ij} = 0$ for all $(i,j)\in{\cal E}$. This corresponds to simple (consensus)  edge constraints of the form $(\forall (i,j)\in {\cal E}) \,\, x_i=x_j$. The results, however, can straightforwardly be generalised to arbitrary dimensions and arbitrary (linear) edge constraints.

The local objective functions $f_i$  are a function of the primal optimisation variable $x_i$, but also of data related to node/agent $i$. 
Examples of such cost functions are $f_i(x_i) = \|x_i-\ms_i\|^2$ in the case of average consensus, where $\ms_i$ is the local measurement data that needs to be averaged among the nodes/agents, or $f_i(x_i) = \|\mA_ix_i-\mb_i\|^2$ in the case of linear regression, where the pair $\ms_i = (\mA_i,\mb_i)$ are the  independent (explanatory) and dependent (response) variables at node $i$. 
In certain applications like medical systems, financial analysis, and smart grids, the data is privacy sensitive and needs to be protected for being revealed to the outside world. In this paper we will focus on solving \eqref{eq:problem}  while protecting the secret/private data $s_i$ from being revealed to the outside world.

\section{Background}
In this section, we formulate the mathematical background of privacy-preserving distributed optimisation and introduce the primal-dual method of multipliers (PDMM) updating equations that will be used in the proof in our work. For more details regarding PDMM we refer to \cite{zha:18,sherson_derivation_2019}.

\subsection{PDMM}
PDMM is an iterative algorithm that can solve the optimisation problem stated in \eqref{eq:problem}. An insightful derivation of the algorithm using monotone operator theory is given in \cite{sherson_derivation_2019}.
As derived in \cite{sherson_derivation_2019}, PDMM can be formulated as a monotone inclusion problem that can be iteratively solved using Peaceman-Rachford splitting (see \cite[Sec. 26.4]{bauschke_convex_2017}). This leads to the following PDMM update equations:
\begin{align}
x^{(k+1)} &=\displaystyle \arg\min_x\left(f(x) + z^{(k)\raisebox{.4mm}{\scriptsize $\mT$\!}}\mC  x + \frac{c}{2}\|\mC  x\|^2\right),  \label{eq:pdmma}\\
z^{(k+1)} &= \mP z^{(k)} + 2c\mP\mC  x^{(k+1)}, \label{eq:pdmmc}
\end{align}
where $x\in\mathbb{R}^{n}$ is the primal variable, 
$z\in\mathbb{R}^{2m}$ an auxiliary variables, and $c>0$ is a constant that controls the convergence rate.
The matrix $\mC \in\mathbb{R}^{2m\times n}$ is a constraint matrix constructed as
$\mC =(c_1,c_2,\ldots, c_{|\mathcal{V}|})$, $c_i\in\mathbb{R}^{2m\times n_i}$, where $c_{i}(l)=\mA_{ij}$  if and only if $e_l = (i,j)\in \cal E$ and $i<j$, and $c_{i}(l+m)=\mA_{ij}$ if and only if $e_l = (i,j)\in \cal E$ and $i>j$. $\mP\in\mathbb{R}^{2m\times 2m}$ is a symmetric permutation matrix exchanging the first $m$ with the last $m$ rows.

Due to the construction of the PDMM algorithm, the update equations are separable across the nodes of the network and thus can be performed in a distributed/parallel manner. 
At every iteration, each node $i$ updates its local optimisation variable $x_i$, and computes, for every neighbouring node  $j \in \mathcal{N}_i$, an auxiliary variable $z_{j|i}$. In particular, \eqref{eq:pdmma}-\eqref{eq:pdmmc} can be computed as
\begin{align}
    &x_i^{(k+1)} = \displaystyle \arg\min_{x_i}\left(\!f_i(x_i) + \!\sum_{j\in{\cal N}_i} \!\!z_{i|j}^{(k)\raisebox{.5mm}{\scriptsize $\mT$\!\!}} \mA_{ij} x_i + \frac{c}{2}d_ix_i^{(k+1)} \!\right)\!\!,
    \label{eq:pdmmai} \\
    &(\forall j\in {\cal N}_i)\quad z_{j|i}^{(k+1)} =  z_{i|j}^{(k)} + 2c\mA_{ij}x_i^{(k+1)}.
    \label{eq:pdmmci}
\end{align}
PDMM converges for strongly convex and differentiable cost functions. This holds for both synchronous \cite{sherson_derivation_2019} and asynchronous \cite{jor:23} update schemes. To guarantee convergence for arbitrary closed, convex and proper (CCP) cost functions, we can apply operator averaging, resulting in the $\theta$-averaged auxiliary vector update 
\begin{equation}\label{eq:zaveraged}
    z_{j|i}^{(k+1)} = (1-\theta)z_{j|i}^{(k)}+\theta \big(z_{i|j}^{(k)} + 2c\mA_{ij}x_i^{(k+1)}\big),
\end{equation}
with $\theta\in(0,1)$.

\subsection{Threat model and privacy metric}

In this paper we will consider two widely-used adversary models: the eavesdropping and the passive (or honest-but-curious) adversary model. The passive adversary consists of a number of colluding nodes, referred to as {\em corrupt nodes}, which will follow the algorithm instructions but will use  received information to infer the input data of the other, non-corrupt nodes. We will refer to the latter as {\em honest nodes}. 
The eavesdropping adversary is usually neglected in existing schemes since eavesdropping can be prevented by using channel encryption \cite{dolev1993perfectly}. However, channel encryption is computationally expensive. For iterative algorithms where the communication channels are used many times, channel encryption is less attractive. We thus assume that the communication in the network is performed through non-secure channels, except for the communication during the initialisation of the network, as will be explained later.

Two commonly used metrics for measuring the privacy loss are $\epsilon$-differential privacy and mutual information \cite{cover2012elements}.  In this paper we choose mutual information as the individual privacy metric. The reason for it is twofold. Firstly,  it has been proven effective in the context of privacy-preserving distributed processing \cite{Jane2020TIFS}, and has been applied in various applications \cite{bu2020tightening}. Secondly, mutual information is closely related to $\epsilon$-differential privacy (see \cite{cuff2016differential} for more details), and is easier to realise in practice \cite{gotz2011publishing}. 
Let $S_i$ be the secret/private continuous random variable that relates to node $i$ having variance $\sigma_{S_i}^2$. Let $Y_i = S_i + N_i$ be what the adversaries can observe related to $S_i$, where $N_i$ is a random variable that is statistically independent of $S_i$, acting as noise obfuscating $S_i$. 
Using $I(S_i; Y_i)$ to denote the mutual information between $S_i$ and $N_i$, we can define $\delta$-level privacy as $I(S_i; Y_i)\leq \delta$.

\section{Privacy preservation}
\label{sec:subspace}

To analyse what information the adversaries can deduce, we can express the primal update equation \eqref{eq:pdmmai} as
\begin{equation}
  0 \in \partial f_i(x_i^{(k+1)}) + \sum_{j\in{\cal N}_i} \mA_{ij} z_{i|j}^{(k)} + cd_ix_i^{(k+1)}.
  \label{eq:opt}
\end{equation}
The private data of the nodes is present in the objective function $f_i$, so the only term in \eqref{eq:opt} that contains private data is $\partial f_i(x_i^{(k+1)})$. Through knowledge of the PDMM algorithm, adversaries can use the
inclusion stated above to deduce information about the term $\partial f_i(x_i^{(k+1)})$. The data that is known to the
adversaries depends on the implementation of PDMM. 

Since the adversary can eavesdrop all communication channels, by inspection of \eqref{eq:opt}, transmitting the auxiliary variables $z_{i|j}^{(k)}$  would reveal $\partial f_i(x_i^{(k+1)})$, as $x_i^{(k+1)}$ can be determined from \eqref{eq:zaveraged}. Encrypting $z_{i|j}^{(k)}$ at every iteration would solve the problem but is computationally too demanding. To overcome this problem, only initial values $z_{i|j}^{(0)}$ are securely transmitted and $\Delta z_{i|j}^{(k+1)} = z_{i|j}^{(k+1)} - z_{i|j}^{(k)}$ are transmitted (without any encryption) during subsequent iterations \cite{Jane2020TSP,jane2021elsevier}. As a consequence,
$z_{i|j}^{(k+1)}$ can only be determined whenever 
$z_{i|j}^{(0)}$ is known.
Let ${\cal N}_{i,h}$ denote the set of honest neighbours of node $i$ and assume that ${\cal N}_{i,h}\neq \emptyset$. That is, we assume that each honest node has at least one honest neighbour. 
After deducing the known terms from \eqref{eq:opt}, the adversaries can observe
\begin{equation}
    \partial f_i(x_i^{(k+1)}) + \sum_{j\in{\cal N}_{i,h}} \mA_{ij} z_{i|j}^{(k)}.
    \label{eq:priv}
\end{equation}

By inspection of \eqref{eq:priv}, we conclude that we can obfuscate the secret/private data $s_i$ by making the variance of $Z_{i|j}$
sufficiently large. This can be achieved by initialising the auxiliary variables $Z_{i|j}^{(0)}$ with sufficiently large variance, assuming this variance remains sufficiently large during the iterations. In \cite{Jane2020ICASSP, Jane2020LS, Jane2020TSP} it has been shown that for a synchronous update scheme this is indeed the case. The reason for this is that at every two synchronous PDMM updates the auxiliary variables are only affected in the subspace $\Psi = {\rm ran}(\mC ) + {\rm ran}(\mP\mC )$ and leave  $\Psi^{\perp} = {\rm ker}(\mC ^\mT ) \cap {\rm ker}((\mP\mC )^\mT )$, the orthogonal complement of $\Psi$, unchanged. 
As shown in \cite{sherson_derivation_2019}, $\Pi_{\Psi} z^{(k)}\to z^*$ for any $z^{(0)}$, where $z^{(k)}_{\Psi} = \Pi_{\Psi} z^{(k)}$ denotes the orthogonal projection of $z^{(k)}$ onto $\Psi$. Because of this the term $\Pi_{\Psi} z^{(k)}$ will eventually always have the same value so that $Z^{(k)}_{\Psi}$ will have zero variance, and cannot be used to obfuscate the private data. For the component $z_{\Psi^\perp}^{(k)}$ in the orthogonal subspace $\Psi^\perp$, however, it
has been shown in \cite{jane2021elsevier} that 
\[
    {\mathbb{E}} \left(Z_{\Psi^\perp}^{(k)} Z_{\Psi^\perp}^{(k)^\mT }\right) = \Pi_{\Psi^\perp}\!\!\left( \frac{\sigma^2}{2} \left((\mI+\mP) + |1-2\theta|^{2k}(\mI-\mP)\right)\!\right).
\]
This expression indicates that the variance of $Z_{\Psi^\perp}^{(k)}$ has a nonzero asymptotic  component that depends on the initialisation variance, and can be made arbitrarily large to obtain $\delta$-level privacy.

For a stochastic update scheme, however, the $z$-updates are not restricted to the subspace $\Psi$ anymore \cite{jor:23}, and deriving a bound for the variance of the auxiliary variable $z$ in this case is more involved than in the synchronous case. In the next section we will derive a lower bound for the variance of $Z_{\Psi^\perp}^{(k)}$  in a stochastic setting\footnote{Since in this framework we transmit differential auxiliary variables $\Delta z_{i|j}$, we assume that there are no transmission losses. With such losses the algorithm as presented here will fail to converge.}. 

\section{Privacy preservation with stochastic PDMM}
\label{sec:bound}

Following \cite{jor:23}, stochastic updates can be defined by assuming that each auxiliary variable $z_{i|j}$ can be updated based on a Bernoulli random variable $U_{i|j}\in\{0,1\}$. 
Collecting all random variables $U_{i|j}$ in the random vector $U\in\mathbb{R}^{2m}$, following the same ordering as the entries of $z$, 
let $(U^{(k)})_{k\in\mathbb{N}}$ denote an i.i.d.\ random process defined on a common probability space $(\Omega,{\cal A})$.
Hence, $U^{(k)}(\omega) \subseteq \{0,1\}^{2m}$ indicates  which entries of $z^{(k)}$ will be updated at iteration $k$,
where we assume that at every iteration, entry $z_{i|j}^{(k)}$ has nonzero probability to be updated. 
With this, a stochastic PDMM iteration can be expressed as
\begin{equation}
    Z^{(k+1)} = \left(\mI-\theta U^{(k+1)}\right)Z^{(k)}+ \theta U^{(k+1)}(\mP Z^{(k)} + 2c\mP\mC  X^{(k+1)}).
    \label{eq:stoch}
\end{equation}
We will utilise the following lemma and some properties of the conditional expectation to derive a lower bound on the variance of $Z_{\Psi^\perp}^{(k)}$ for stochastic PDMM.

\begin{lemma}
Let $\Pi_{\Psi^\perp}$ denote the projection onto the subspace $\Psi^\perp$ and $\mP $ denote the PDMM permutation matrix. Then 
\begin{itemize}
    \item[\rm a)] $\displaystyle \Pi_{\Psi^\perp}\mP = \mP\Pi_{\Psi^\perp}$,
    \item[\rm b)] $\displaystyle \Pi_{\Psi^\perp} \mathbb{E}(X) = \mathbb{E} \big( \Pi_{\Psi^\perp} X\big)$,
    \item[\rm c)] $\displaystyle \Pi_{\Psi^\perp} \big( \mP\mC  x^{(k+1)} \big) = 0$.
\end{itemize}
\label{lem:1}
\end{lemma}
\begin{proof}
For a), see \cite[Lemma 5.2]{jane2021elsevier}. Item b) follows from the linearity of the expectation operator. Item c) follows from the fact that $\mP \mC  x^{(k+1)} \in \Psi$.
\end{proof}
To simplify the analysis, we make use of the observation that in expected value stochastic PDMM with uniform updating probabilities, that is, $(\forall i\in {\cal V}) (\forall j\in{\cal N}_i) \,\, {\rm Pr}(\{U_{i|j}^{(k)} = 1 \})=\mu$, is equivalent to $\mu$-averaged PDMM. 
This assumption holds, for example, in the case of asynchronous PDMM where one node is selected uniformly at random and updated each iteration, so that $\mu = 1/|\cal V|$. 

Let $({\cal A}_k)_{k\geq 1}$ be a filtration on $(\Omega,{\cal A})$ such that
${\cal A}_k := \sigma\{U^{(m)}\,:\, m\leq k\}$,
the $\sigma$-algebra generated by the random variables $U^{(1)},\ldots,U^{(m)}$ and thus ${\cal A}_k \subseteq {\cal A}_l$ for $k\leq l$. 
By taking the conditional expectation of \eqref{eq:stoch} with respect to ${\cal A}_0$, and using the fact that $U^{(k+1)}$ is independent of $Z^{(k)}$ and $X^{(k+1)}$, we have that
\begin{align*}
\mathbb{E}(Z^{(k+1)} | {\cal A}_0)  &= \left(\mI-\theta \mathbb{E}(U^{(k+1)})\right)\mathbb{E}(Z^{(k)} | {\cal A}_0) + \nonumber \\
&\hspace{-5mm}\theta\mathbb{E}(U^{(k+1)})\!\left(\mP \mathbb{E}(Z^{(k)} | {\cal A}_0) + 2c\mP\mC  \mathbb{E}(X^{(k+1)} | {\cal A}_0)\right) \nonumber \\
&= (1-\theta\mu) \mathbb{E}(Z^{(k)} | {\cal A}_0) + \nonumber \\
&\hspace{-5mm}\theta\mu\!\left(\mP \mathbb{E}(Z^{(k)} | {\cal A}_0) + 2c\mP\mC  \mathbb{E}(X^{(k+1)} | {\cal A}_0)\right),
\end{align*}
where for the last step we assume equal updating probabilities. As discussed in Section~\ref{sec:subspace}, the part of $Z$ in the subspace $\Psi^\perp$ is relevant for privacy preservation.
Hence, we have by Lemma~\ref{lem:1} that 
\begin{align}
\mathbb{E}(Z_{\Psi^\perp}^{(k)} | {\cal A}_0) 
&=  \left((1-\theta\mu) \mI + \theta \mu \mP\right) \mathbb{E}(Z_{\Psi^\perp}^{(k-1)} | {\cal A}_0) \nonumber \\
&= \left((1-\theta\mu) \mI + \theta \mu \mP\right)^{k} Z_{\Psi^\perp}^{(0)}\nonumber \\
&= \frac{1}{2}\left( (\mI+\mP) + (1-2\theta\mu)^{k} (\mI-\mP)\right)Z_{\Psi^\perp}^{(0)},
\label{eq:EZA}
\end{align}
where the last equality is shown in \cite{jane2021elsevier} using the eigenvalue decomposition of $(1-\theta\mu) \mI + \theta \mu \mP$. Hence, we have
\begin{align*}
\mathbb{E} \left(Z_{\Psi^\perp}^{(k)} Z_{\Psi^\perp}^{(k)^\mT }\right)_{(i|j),(i|j)} &= \mathbb{E} \left(\left(Z_{\Psi^\perp}^{(k)}\right)_{(i|j)}^2\right) \\
&= \mathbb{E} \left( \mathbb{E} \left( \left(Z_{\Psi^\perp}^{(k)}\right)_{(i|j)}^2 |{\cal A}_0 \right) \right) \\
&\geq \mathbb{E} \left( \mathbb{E} \left( \left(Z_{\Psi^\perp}^{(k)}\right)_{(i|j)} |{\cal A}_0 \right)^2 \right),
\end{align*}
where the last inequality follows from Jensen's inequality. Hence, if the auxiliary
variable is randomly initialised in such a way that $\mathbb{E}(Z^{(0)} Z^{(0)^\mT }) = \sigma^2 \mI$, we have that
\begin{align*}
{\rm diag}\!\left( \mathbb{E} \left(Z_{\Psi^\perp}^{(k)} Z_{\Psi^\perp}^{(k)^\mT }\right) \right) &\\
&\hspace{-20mm}\geq  {\rm diag} \!\left( \mathbb{E} \left( \mathbb{E} \left( Z_{\Psi^\perp}^{(k)} |{\cal A}_0 \right) \mathbb{E} \left( Z_{\Psi^\perp}^{(k)} |{\cal A}_0 \right)^{\!\mT}\right) \!\right) \\
&\hspace{-20mm}= {\rm diag} \! \left( \Pi_{\Psi^\perp}\!\!\left( \frac{\sigma^2}{2} \left((\mI+\mP) + |1-2\theta\mu|^{2k}(\mI-\mP)\right)\!\right)\!\right),
\end{align*}
where the last equality follows from substitution of \eqref{eq:EZA}, Lemma~\ref{lem:1}, and the fact that $\Pi_{\Psi^\perp}\Pi_{\Psi^\perp}^\mT = \Pi_{\Psi^\perp}$. This inequality provides a lower bound for the variance of the auxiliary variable that is based on the initialisation variance $\sigma^2$. Note that this bound is equal to the variance of the auxiliary variables for synchronous $\theta\mu$-averaged PDMM. Since for most practical networks $\theta\mu\ll 1$,
 the variance will be lower bounded by the variance of synchronous PDMM.

\begin{figure}[t]
\centering
\includegraphics[width=.45\textwidth]{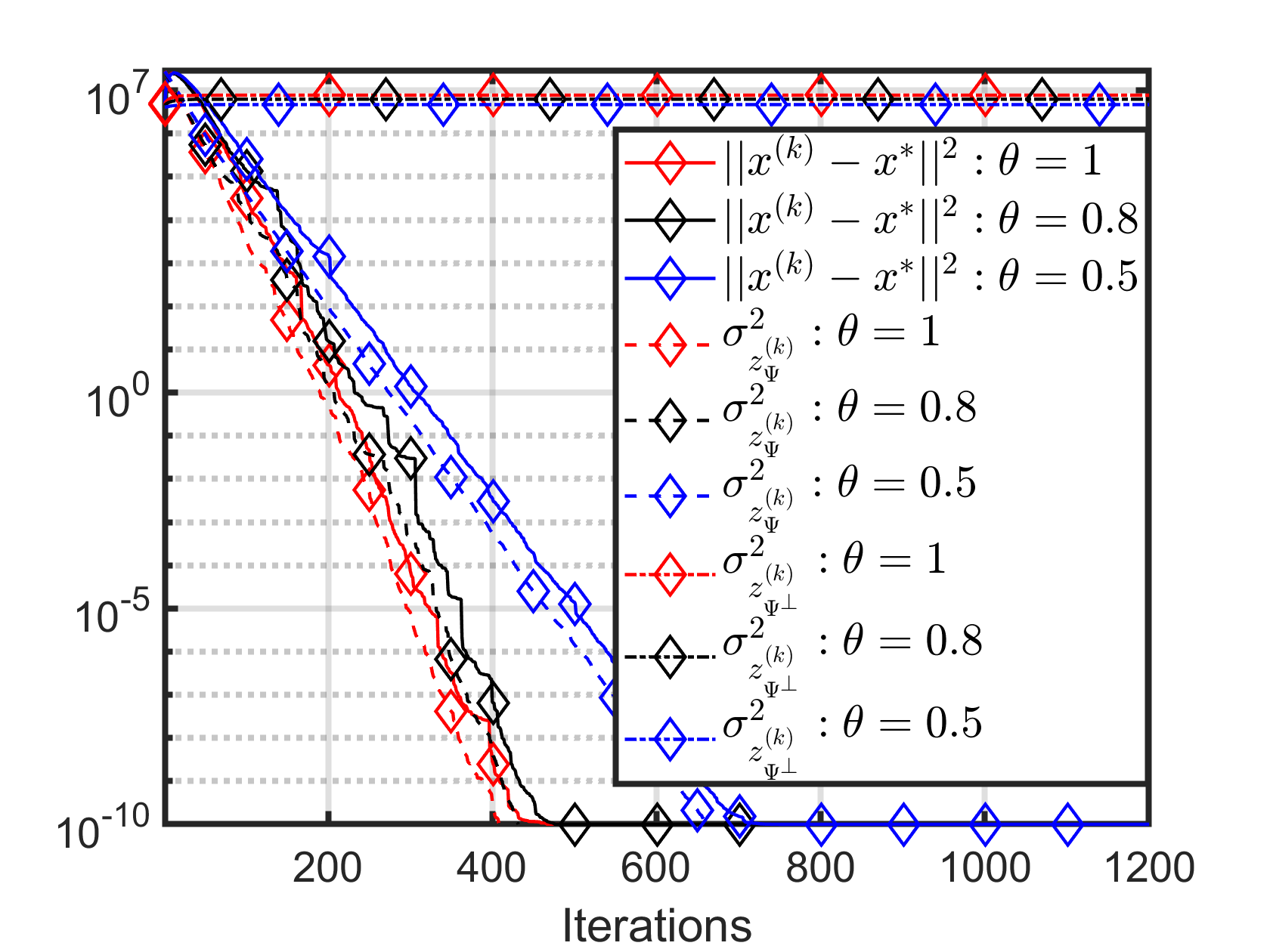}
 \centering
 \mbox{}\vspace*{-.8\baselineskip}
\caption{Convergence of the variables $x^{(k)}$, $Z_{\Psi}^{(k)}$, and $Z_{\Psi^\perp}^{(k)}$
using asynchronous PDMM for $\theta=1, 0.8$, and $0.5$. }
\label{fig:x}
\vskip -10pt
\end{figure}
\section{Numerical results}
To validate the variance bound derived in Section~\ref{sec:bound}, we perform distributed averaging simulations using asynchronous PDMM over a random geometric graph with $n=10$ nodes having a communication radius $r = \sqrt{2\log n/n}$
so that the graph is connected with probability at least $1-1/n^2$ \cite{dall2002random}.  
The nodes are being activated each iteration with uniform probability $\mu = 1/|\cal V|$.
The values in the private measurement data $s$ were randomly generated from a zero-mean unit-variance Gaussian distribution and the
auxiliary data $Z^{(0)}$ was initialized with zero mean and $\sigma^2_{Z^{(0)}} = 10^8$. To analyse the variance of $Z^{(k)}$, we performed 100 Monte Carlo runs with different initialisation $z^{(0)}$ where, for consistency, the sequence of selected nodes was kept constant for each Monte Carlo run. 

Figure~\ref{fig:x} shows convergence results of both the primal variable $x$ and the auxiliary variable $z$  as a function of the number of iterations for different values of $\theta$. The primal variable $x^{(k)}$ converges to the correct solution and has, as expected, a decreasing convergence rate for increasing values of $\theta$.  As for the auxiliary variable $z$,
 we can see that the variance of $Z_{\Psi}^{(k)}$ goes to zero as $k\to\infty$, whereas the variance of $Z_{\Psi^\perp}^{(k)}$ remains nonzero. Figure~\ref{fig:compare} shows the 
evaluation of $Z_{\Psi^\perp}^{(k)}$ in more detail and shows a comparison with the results obtained with synchronous PDMM \cite{jane2021elsevier}.
The results show that in both asynchronous and synchronous PDMM the variance of $Z_{\Psi^\perp}^{(k)}$ is bounded below by a value depending on the variance of the initialisation $Z_{\Psi^\perp}^{(0)}$, and that this variance is higher in the case of an asynchronous update scheme.

Figure~\ref{fig:mi} shows the normalized mutual information based on \eqref{eq:priv} by assuming there is one honest neighbor $j$ in the neighborhood of node $i$ for different values of the variance of
$Z^{(0)}$. The results show that the larger  $\sigma^2_{Z^{(0)}}$ is, the smaller the mutual information is between the private data and the observations of the adversary. It also confirms the fact that asynchronous update schemes reveal less information about the private data during the iterations than synchronous schemes.
\vspace{-.2\baselineskip}
\begin{figure}[t]
\centering
\includegraphics[width=.42\textwidth]{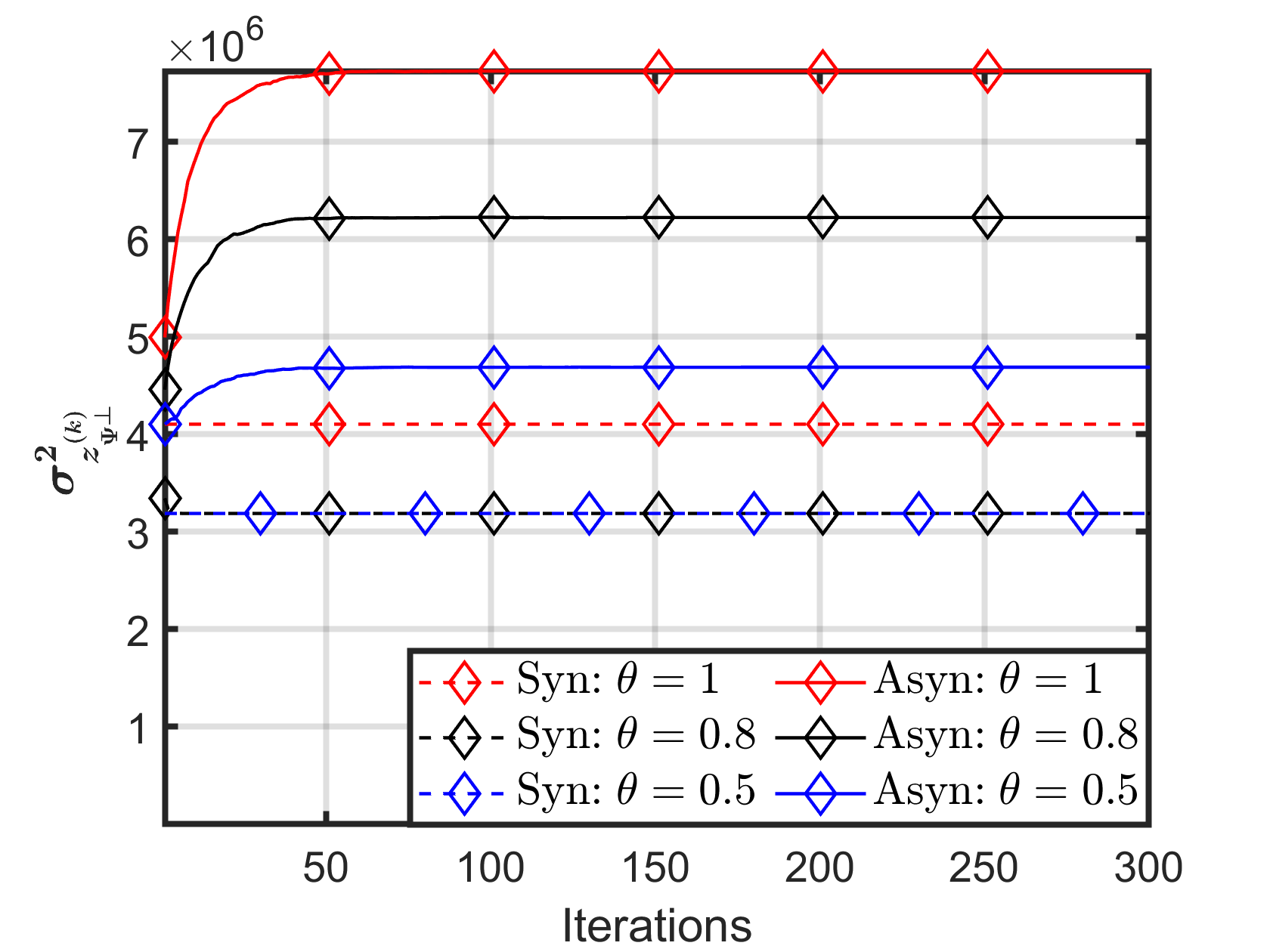}
 \centering
 \mbox{}\vspace*{-.3\baselineskip}
\caption{Variance of $Z_{\Psi^\perp}^{(k)}$ using synchronous and asynchronous PDMM 
for $\theta=1, 0.8$, and  $0.5$. }
\label{fig:compare}
\vskip -5pt
\end{figure}
\begin{figure}[t]
\centering
\includegraphics[width=.42\textwidth]{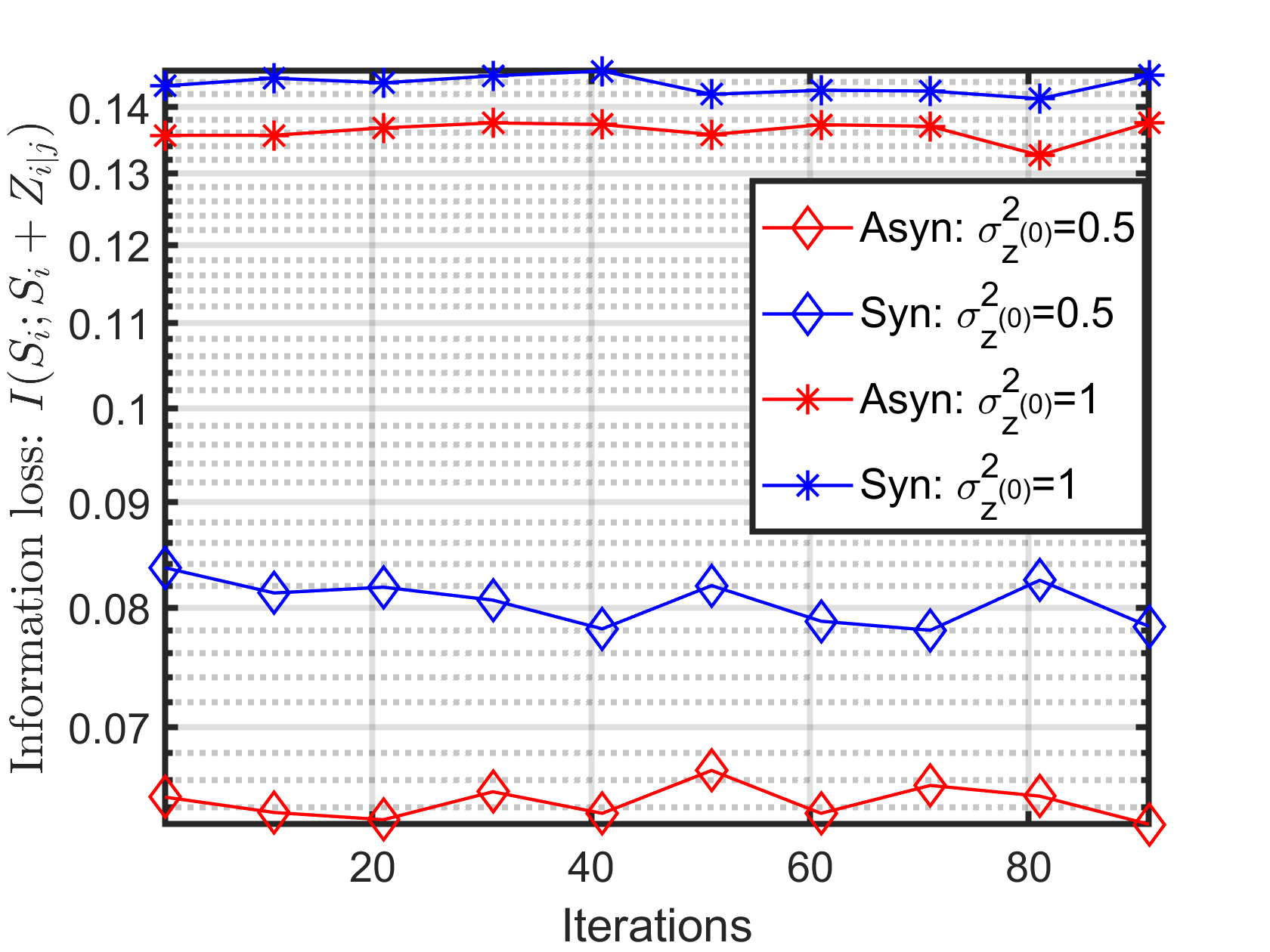}
 \centering
 \mbox{}\vspace*{-.8\baselineskip}
\caption{Information loss in asynchronous and synchronous PDMM for $\sigma^2_{Z^{(0)}}=0.5$ and $\sigma^2_{Z^{(0)}}=1$.   }
\label{fig:mi}
\vskip -10pt
\end{figure}


\vspace{-.6\baselineskip}
\section{Conclusion}

In this paper we investigated the privacy preservation of asynchronous distributed optimisation schemes. We extended the results of subspace perturbation, originally introduced for synchronous schemes, to stochastic update schemes. We derived a lower bound on the noise variance that can be used to obfuscate the private data when the probability of a node being updated is uniform. The bound guarantees that the obtained privacy level of asynchronous schemes will be higher than the original subspace perturbation approach, while the algorithm's accuracy still remains uncompromised. Computer simulations confirmed our theoretical findings.

\newpage
\bibliographystyle{IEEEbib}
\bibliography{ref}

\end{document}